\documentclass[journal]{IEEEtran}
\usepackage{lineno,hyperref}
\usepackage{amsmath,amssymb,amsfonts}
\usepackage{algorithm}
\usepackage{algorithmicx}
\usepackage{algpseudocode}
\usepackage{graphicx}
\usepackage{epstopdf}
\usepackage{textcomp}
\usepackage{xcolor}
\usepackage{mathtools}
\usepackage{amsthm}
\usepackage{url}
\usepackage{mathrsfs}
\usepackage{subfigure}
\usepackage[misc]{ifsym}
\usepackage[justification=centering]{caption}	

\newtheorem{myDef}{Definition}

\newtheorem{theorem}{Theorem}
\newtheorem{lemma}{Lemma}

\algnewcommand\Input{\item[\textbf{Input:}]}
\algnewcommand\Output{\item[\textbf{Output:}]}

\ifCLASSINFOpdf
\else
\fi
\begin{document}
\title{Continuous Influence-based Community Partition for Social Networks}

\author{Qiufen Ni,
        Jianxiong Guo,
        Weili Wu,~\IEEEmembership{Senior Memeber,~IEEE},
        and Chuanhe Huang % <-this % stops a space

\thanks{This work has been submitted to the IEEE for possible publication. Copyright may be transferred without notice, after which this version may no longer be accessible.}
\thanks{Qiufen Ni (\textit{Corresponding author}) is with School of Computers, Guangdong University of Technology, Guangzhou 510006, China.}% <-this % stops a space
\thanks{Jianxiong Guo and Weili Wu are with Department of Computer Science, The University of Texas at Dallas, Richardson, Texas 75080, USA.}
\thanks{Chuanhe Huang is with School of Computer Science, Wuhan University, Wuhan 430072, China.}}


\markboth{Journ of \LaTeX\ Cla Fil,~Vo.~, N.~, August~2020}%
{Shell \MakeLowercase{\textit{et al.}}: Bare Demo of IEEEtran.cls for IEEE Journals}



\maketitle

\begin{abstract}
  Community partition is of great importance in social networks because of the rapid increasing network scale, data and applications. We consider the community partition problem under LT model in social networks, which is a combinatorial optimization problem that divides the social network to disjoint $m$ communities. Our goal is to maximize the sum of influence propagation through maximizing it within each community.  As the influence propagation function of community partition problem is supermodular under LT model, we use the method of Lov{$\acute{a}$}sz Extension to relax the target influence function and transfer our goal to maximize the relaxed function over a matroid polytope. Next, we propose a continuous greedy algorithm using the properties of the relaxed function to solve our problem, which needs to be discretized in concrete implementation. Then, random rounding technique is used to convert the fractional solution to integer solution. We present a theoretical analysis with $1-1/e$ approximation ratio for the proposed algorithms. Extensive experiments are conducted to evaluate the performance of the proposed continuous greedy algorithms on real-world online social networks datasets and the results demonstrate that continuous community partition method can improve influence spread and accuracy of the community partition effectively.
\end{abstract}


% \fi
\IEEEpeerreviewmaketitle

\section{Introduction}
{N}owadays, social networks have gained popularity in real-world applications. So its structure becomes more and more complex and enormously large. It may cause the network traffic congestion as the increasing of the social network users and may increase the communication cost as the data storing in different servers, it may also brings difficulty in dealing with and analysing the complex data sources. A social network can be seen as a graph of social individuals and relationship between them, where the social individuals represent the vertices and the connections between individuals constitute the edges of the graph \cite{ni2020information}. A valuable approach in analysing large complex social networks is community partition. A community is a group of nodes with dense inner connections and relatively sparse connections with nodes outside the group. Community partition is to partition a complex social network into several medium size sub-networks communities, which can help us learn about the relationships and characteristics of the individuals more clearly and improve the process of analysis since each community has similar node distribution with the original complex social networks. Some applications, such as the community-based rumor blocking, community-based active friending and so on, take advantage of the community structure of social networks to help us solve problems more effectively.

Influence maximization is an important problem in social networks. Its target is to select a small set of seed users to trigger a large number of influence propagation, which is widely used in the viral marketing \cite{banerjee2019maximizing, liu2018active, talukder2018cost}. As the uncertainty of human behaviours and decision making affects each other, Kemp {\em et al.}\cite{kempe2003maximizing} proposed two classic influence propagation probabilistic models: Linear  Threshold (LT) model and Independent Cascade (IC) model. They proved that the expected number of active users, called {\em influence spread}, is monotone and submodular with respect to the seed set. They propose a greedy algorithm which can maximize the influence propagation in social networks.  Lu {\em et al.} \cite{lu2017solution}  and Wang {\em et al.} \cite{wang2016list} show that the influence maximization problem is NP-hard under IC model while polynomial-time solvable in LT model. In this paper, we want to study the community partition problem with the goal of maximizing the influence propagation in each community under LT model.

Community partition attracts extensive attentions. To solve it, most of community partition algorithms consider the social network as a graph where each node only belongs to one community \cite{huang2019community}. The second approach is to correlate the social network with a hypergraph that nodes overlap exist between communities, there are some related studies\cite{yang2018hepart}. The third approach associates the concept graph or Galois lattices where nodes share some common characteristics or knowledge, which is a more complex structure than the first two categories and it can give more semantics of communities to the network structure. The output is a Galos hierarchy hypergraph marked with lattice intents\cite{plantie2010photo, fazlali2017adaptive}. In this work, we investigate the influence-based disjoint community partition problem in a graph, as social networks with disjoint community structures are real, such as communities composed of experts in different fields, a work unit’s departmental organization and personnel division communities, etc., this research problem has its practical significance.
We adopt the Lov{$\acute{a}$}sz extension to relax the objective function and design a continuous greedy algorithm to partition a social network into $m$ disjoint communities with the goal of maximizing the influence propagation within each community. Different from the existing heuristic methods in community partition problem, the proposed algorithm ensures the accuracy in community partition and can achieve a good approximation theoretical guarantee at the same time. Besides the community partition, there are also some other works study the local influence propagation instead of the global one\cite{bartal2019modeling, chen2010scalable1}.

The main contributions of this paper can be summarize as follows:

\begin{itemize}
	\item An innovative continuous influence-based community partition method is developed. First, we formulate the community partition problem as partitioning a social network to $m$ disjoint communities with the goal of maximizing the influence propagation within each community.
	\item We use the Lov{$\acute{a}$}sz Extension to relax the objective function and a partition matroid to the domain of the relaxed problem is introduced.
	\item A continuous greedy algorithm is devised based on its continuous extension and its discrete form is proposed to solve the problem in concrete implementation.
	\item We analyze the performance guarantee and get an approximation ratio $1-1/e$ for the proposed algorithms.
	\item We do simulations on three real-world online social networks datasets to verify the high-quality and accuracy of the proposed continuous algorithms.
\end{itemize}

The result of this paper is arranged as follows.
In Section \ref{related-work} we start with summarizing some existing related work. Then in Section \ref{network-model}, the network model and problem formulation are introduced. Section \ref{solution} gives the detail solution for the proposed community partition problem. We also give the theoretical proof of the proposed algorithm in Section \ref{Performance}, and in Section \ref{experiments} the simulation results are presented, while finally, the conclusion is presented in Section \ref{conclusion}.
\section{Related Work}\label{related-work}
Community partition is one kind of community detection method. The community detection is not only important in social networks, but also widely studied in other fields, like biological networks \cite{zapp2019mechanoradicals, al2019biomolecular} and technological networks \cite{kaminski2019clustering, kishore2019nature}. Most of previous work about community detection is presented from the perspective of network structure: (1) Hierarchy-based method. M. Girvan {\em et al.} \cite{girvan2002community} propose a method which can be used for community detection in social and biological networks. They present a greedy algorithm progressively which removes the edge with the most betweenness communities from the network graph. (2) Modularity-based method. In this method, each node is an independent community initially. But if merging two communities can gain a larger modularity, the merging process will go on until the modularity is stable. Newman {\em et al.} \cite{newman2004fast} define a modularity $Q$, then they divide the network into communities based on the value of $Q$, when $Q$ reduces to zero, it represents that no more interior community edges than would be expected by random chance. But these algorithms only limit applying to small size networks. (3) Spectrum-based method. This method is based on the multi-path partition to achieve the spectral clustering. In order to gain overlapping communities, C. Gui {\em et al.} \cite{gui2018overlapping} propose a new algorithm to establish hierarchical structure with detecting the overlap of communities. This algorithm can balance the overlap and hierarchy through spectral analysis method. (4) Dynamic-based method. It dynamically chooses nodes with more gains of community measure function to partition the community. Q. Liu {\em et al.} \cite{liu2017evolutionary} design a algorithm to discovery the link community structure of a dynamic weighted network based on the  fitness of weighted edges and partition density. This method can show the process of the evolution of the link community structure, and it can also detect the overlapping communities. Through adjusting the value of a parameter, we can get the link communities' hierarchical structure. (5) Label propagation method. It is a local-based community detection method depends on the label propagation of nodes. UN Raghavan {\em et al.} \cite{raghavan2007near} assume that every node in the network is allocated a unique label. Each node iteratively accepts the label which is adopted by most of its neighbours until all the nodes' labels achieve a stable state. Then nodes with the same labels will be partitioned to the same community.

There are some community detection works based on the node attributes. The user's topic, tag and behaviour information of a node can be looked as a node's attributes. We can partition community based on these attribute characteristics. X. Teng {\em et al.} \cite{teng2019overlapping} propose a overlapping community detection algorithm which is based on the nodes' attributes in attribute networks. They design two objects: one is the changed extended modularity value and the other is the attribute similarity value. Then they use an encode and decode approach to realize the overlapping community partition.  X. Wang {\em et al.} \cite{wang2015detecting} consider the group feature instead of the individual characteristics in mostly previous works. They classify the ``Internet water army'' to six communities based on their behaviours and design a community detection method which is based on the logistic regression model.

There are also some other classical community partition methods in complex networks, such as based on multi-objective evolutionary algorithm, algorithm based on local expansion. In  \cite{zou2019inverse}, Z. Feng {\em et al.} propose a discrete Gaussian process-based inverse model and a multi-objective optimization method to solve community detection problem in complex networks. Manuel Guerrero {\em et al.} \cite{guerrero2020multi} study the multi-objective variants of community detection problem, the variants includes modularity, conductance metric, the imbalance of the nodes in communities. They propose a Pareto-based multi-objective evolutionary algorithm to optimize different objectives simultaneously. Kamal Berahmand {\em et al.} \cite{berahmand2018community} propose an expansion and detection of core nodes based local approach, the edge weight which is detected based on the node similarity is more accurate, the local algorithm would be more precise. The proposed algorithm can detect all the graph's communities in a network by local and identifying various nodes' roles.  X. Ding {\em et al.} \cite{ding2020node} propose a overlapping community detection algorithm to identify different community structures, which performs local expansion method and boundary re-checking sub-processes in order.

In recent years, some community detection methods combined with the new technologies have emerged. In \cite{2018An}, W. Feifan {\em et al.} propose an extreme learning machine-based community detection algorithm in complex networks, which combines $k$-means and unsupervised ELM. This proposed algorithm is verified to be outstanding in low complexity. G. Sperlf \cite{10.1145/3297280.3297574} proposes a deep learning and network's topology based community detection algorithm. Convolutional neutral network is used to deal with the large dimensions of a social network. Z. Chen {\em et al.} \cite{2017Supervised} propose a novel community detection method with Graph Neural Networks (GNNs) under the supervised learning setting. The experimental results and analyses demonstrate that the loss of local minima is low in linear GNN models.

There are also some community detection works based on influence propagation. This community partition method is based on which community nodes can achieve larger influence. Y. Wang {\em et al.} \cite{wang2010community} present a selecting community approach based on the node's influence in Mobile Social Networks (MSNs), which mining the top-$K$ most influential nodes within each community. Putting them together then as the best $K$ nodes. Z. Lu {\em et al.} \cite{lu2014influence} consider the community partition problem as a combinatorial optimization problem, which dividing $K$ disjoint communities in the target of maximizing the influence in each community. Then they propose a $MKCP$ algorithm to solve this problem. In particular, when $K=2$, they develop an optimal algorithm to partition the social networks to two disjoint sub-networks. N. Barbieri {\em et al.} \cite{barbieri2016efficient} define a stochastic framework to model the social influence of users within a community, and they present an expectation maximization (EM) learning algorithm, which allows the automatic detection of the most appropriate number of communities by enabling a community annihilation mechanism. All these existing influence-based community partition methods are discrete and heuristic, they also do not have a theoretical guarantee.

\section{Network Model and Problem Formulation}\label{network-model}
\subsection{The Network Model}
A social network is modelled as a directed graph $G=(V, E)$, where each vertex $i$ in $V$ is an individual, and each edge $e=(i, j)$ in $E$ is the social tie between user $i$ and $j$. Let $ N^-(i)$ and $N^+(i)$ denote the sets of incoming neighbours and outgoing neighbours, respectively. In LT model, each edge $e \in E$ in the graph is associated with a weight $w_{ij}$, each node $i\in V$ is influenced by its incoming neighbours $j$ satisfies $\sum_{j \in N^-(i)}w_{ij}\leq1$. In addition, each node $i \in V$ is related with a threshold $\theta_i$ which is uniformly distributed in the interval $[0,1]$. The information diffusion process can be described in discrete steps: all nodes that are active in step $t-1$ will still active in step $t$. An inactive node will be active if the total weight of its incoming neighbours that are active is larger than or equal to $\theta_i$, i.e. $\sum_{j \in N^-(i)}w_{ij}\geq\theta_i$. The propagation process ends until there is no new node being activated.

\subsection{Problem Formulation}
Assume that there are $m$ communities, we allocate a community identifier $s_j\in M= \{1,2,...,m\}$ for each node $j$, so all the nodes in the same community have the same community identifier, i.e. $S_i=\{j|s_j=i\}$ represents the node set in community $S_i$, where $1\leq i\leq m$. For a  community $S_k$ and a node $i\in S_k$, we use $\sigma_{S_k}(i)=\sum_{j\in(S_k\backslash i)}p_{S_k}(i, j)$ to denote the influence propagation of node $i$ within community $S_k$. Assume there is a non-empty subset $D\subseteq S_k$, the total influence propagation of all nodes in $D$ within community $S_k$ is denoted by $\sigma_{S_k}(D)=\sum_{i\in D}\sigma_{S_k}(i)$, which can show the reciprocal influence strength among the nodes in community $S_k$. In the rest of the paper, we use $\sigma(Y)$ to  denote the influence propagation of community $Y$ instead of $\sigma_Y(Y)$ for simplicity. So we denote the total influence propagation in the social networks after partitioning to $m$ communities as $f(S_1,S_2,...,S_m)=\sum_{i=1}^{m}\sigma(S_i)$. Next, let us describe the community partition problem we want to solve as follows:

\textbf{Influence Maximization for  Community Partition Problem (IMCPP):} Given a graph $G=(V,E)$ as a social network, its information diffusion is under LT model. We seek a partition of the social network into $m$ disjoint sets \{$S_1, S_2,\dots,S_m$\} satisfying: (1) $\bigcup_{k=1}^{m}(S_k)=V$; (2) ${\forall i\neq j, S_i\cap S_j=\emptyset}$. Our goal is to maximize the influence propagation function $f(S_1,S_2,\dots,S_m)=\sum\limits_{k=1}^{m}\sigma(S_k)$.

Z. Lu {\em et al.} \cite{lu2014influence} proved that the maximum K-community partition problem is NP-hard. Our IMCPP can be reduced to K-community partition problem, thus, the IMCPP is NP-hard.

\section{Solution for IMCPP}\label{solution}
We will prove some important properties of the objective function and how to solve it step by step efficiently in this section.

\subsection{Property of Influence Propagation Function $f$}
First, we need to analyse the properties of the influence propagation function $f$ we want to solve. The first property of $f$ is monotonicity.

\begin{lemma}
	The influence propagation function $f$ for the community partition problem is monotone under the LT model.
	\label{lemma1}
\end{lemma}
\begin{proof}
	Assuming that the influence propagation function within the community $S_k$ is $\sigma(S_k)$. We know that when adding a seed node $i$ to this community, the conditional expected marginal gain produced by $i$ to the community $S_k$ can be denoted as: $\Delta(i|S_k) =\mathbb{E}[\sigma(S_k+i)-\sigma(S_k)]$. Obviously, $\Delta(i|S_k)\geq0$. So $\sigma(S_k)$ is monotone. As we know that the influence propagation function $f$ is: $f(S_1,S_2,\dots,S_m)=\sum_{k=1}^{m}\sigma(S_k)$. Hence, $f$ is also monotone.
\end{proof}
We need to know the definition of supermodular function before we introduce the second property of $f$. Let $X$ with $|X|=n$ be a \textit{ground set}. A set function on $X$ is a function $h$: $2^X\rightarrow R$.
\begin{myDef}[Supermodular function]A set function $h$: $2^X\rightarrow R$ is supermodular if for any $A\subseteq B\subseteq X$ and $u\in X\backslash B$, we have $h(A\cup\{u\})-h(A)\leq h(B\cup\{u\})-h(B)$. There is another equivalent definition for supermodularity, that is $h(A\cap B)-h(A\cup B)\geq h(A)+h(B)$.
\end{myDef}
We found that the influence propagation function $f$ satisfies the supermodularity, shown as the following lemma, that is,
\begin{lemma}
	The influence propagation function $f$ for the community partition problem is supermodular under the LT model.
	\label{lemma2}
\end{lemma}
\begin{proof}
	Assume there are two communities $S_a$ and $S_b$, and $S_a \subset S_b$, so we have to prove that for any node $q\notin S_b$, $\sigma(S_a\cup\{q\})-\sigma(S_a)\leq \sigma(S_b\cup\{q\})-\sigma({S_b})$, this is the condition that a function is supermodular. If there is a live-edge path from seed node $i$ to $j$, it indicates that node $j$ is influenced by seed $i$. Let $p_{S_a}(i,j)$ be the  probability that node $j$ receives influence from node $i$ through nodes within community $S_a$ and $p_{S_b}(i,j)$ be the  probability that node $j$ receives influence from node $i$ through nodes within community $S_b$. So we have: $\sigma(S_a\cup\{q\})-\sigma(S_a)=\sum_{j\in S_a}p_{S_a}(q,j)+\sum_{i\in S_a}p_{S_a}(i,q)+\sum_{i,j\in S_a:i\neq j}\{p_{S_a\cup\{q\}}(i,j)-p_{S_a}(i,j)\}$, this is the sum of the probabilities that  the path must pass $q$ one time in community $(S_a\cup\{q\})$. Accordingly, we can calculate the sum of the probabilities that the path must pass $q$ one time in community $(S_b\cup\{q\})$ as $\sigma(S_b\cup\{q\})-\sigma(S_b)=\sum_{j\in S_b}p_{S_b}(q,j)+\sum_{i\in S_b}p_{S_b}(i,q)+\sum_{i,j\in S_b:i\neq j}\{p_{S_b\cup\{q\}}(i,j)-p_{S_b}(i,j)\}$. As we know that $S_a \subset S_b$, we can get that $\sum_{j\in S_a}p_{S_a}(q,j)\leq\sum_{j\in S_b}p_{S_b}(q,j)$, $\sum_{i\in S_a}p_{S_a}(i,q)\leq\sum_{i\in S_b}p_{S_b}(i,q)$, because $\{S_a\cup\{q\}\}$ is also the subset of  $\{S_b\cup\{q\}\}$. It also follows that $p_{S_a\cup\{q\}}(i,j)-p_{S_a}(i,j)\leq p_{S_b\cup\{q\}}(i,j)-p_{S_b}(i,j)$. So we can get the inequality $\sigma(S_a\cup\{q\})-\sigma(S_a)\leq \sigma(S_b\cup\{q\}-\sigma({S_b})$. Therefore, the influence propagation function $\sigma$ within each community is supermodular under LT model. As $f(S_1,S_2,\dots,S_m)=\sum_{k=1}^{m}\sigma(S_k)$. Therefore, the sum influence propagation function $f$ in a social network for the community partition problem under LT model is also supermodular.
\end{proof}

\subsection{Reformulation of the IMCPP}
First, we need to introduce some basic definitions about matroid and matroid polytopes which will be used later.
\begin{myDef} [Matroid polytopes]Given a matroid $\mathcal{M}=(X,\mathcal{I})$, the matroid polytope $P(\mathcal{M})$ is the convex hull of the indicators of the bases of $\mathcal{M}$ and defined as:
	\begin{equation*}
	P(\mathcal{M})=conv\{\vec{1}_I: I\in \mathcal{I}\}.
	\label{equation11}
	\end{equation*}
	$\mathcal{I}$ is a family of subsets of ground set $X$ (called independent sets).
\end{myDef}
The matroid polytopes $P(\mathcal{M})$ is down-monotone because it satisfies the property that for any $0\leq x\leq y, y\in P\Rightarrow x\in P$.

Then, we generalize the IMCPP problem to a matroid constraint, which is easier to be solved. Here, we define a new ground set $U= M\times V$, where $M$ is the community set and $V$ is the node set of the given graph. Let $A\subseteq U$ be a feasible solution, namely a feasible community partition combination. Here, $(i, j)\in A$ means that we partition the node $j$ to community $i$. As we can not partition the same node to more than one community, thus, a feasible solution satisfies the following constraint, that is
\begin{equation*}
\forall j\in V, |\{i|(i,j)\in A\}|\leq 1
\label{equation1}
\end{equation*}
Then the influence function of a partition $A$ can be denoted as:
\begin{equation*}
f(A)=\sum\limits_{i\in M}\sigma(\{j|(i,j)\in A\})
\end{equation*}
Thus, the IMCPP can be  written as follows:

\begin{equation}
\begin{split}
&\quad\max\limits_{A\subseteq U}f(A)\\
&\quad s.t. \  \forall j\in V, |\{i|(i,j)\in A\}|\leq 1
\end{split}
\label{equation2}
\end{equation}
Therefore, let us define a partition matroid $\mathcal{M}=(U,\mathcal{I})$ as follows:
\begin{equation*}
\mathcal{I}=\{X\subseteq U:|X\cap(M\times \{j\})|\leq 1 \text{ for } j\in V\}
\end{equation*}
Then the IMCPP problem is equivalent to maximize $\{f(A):A\in\mathcal{I}\}$. Any set $A\in\mathcal{I}$ is called independent set.

\subsection{Relaxation of IMCPP}
In this section, we give a continuous relaxation of our optimization problem, shown as Equation \ref{equation2}. First, we need to introduce a continuous extension for an arbitrary set function: Lov{$\acute{a}$}sz extension. It was defined by Lov{$\acute{a}$}sz in \cite{lovasz1983submodular}.
\begin{myDef}[Lov{$\acute{a}$}sz extension]
	For a  function $h$: $2^X\rightarrow R$, $\vec x\in R^X$. Assume that the elements in ground set $X=\{v_1,v_2,\cdots,v_n\}$ are sorted from maximum to minimum such that $x_1\geq x_2\geq\cdots\geq x_n$. Let $S_i=\{v_1,\cdots,v_i\}, \forall v_i\in X$.  The Lov{$\acute{a}$}sz Extension 	$\hat{h}(\vec {x}):[0,1]^X\rightarrow R$ of $h$ at $\vec x$ is defined as:
	\begin{equation*}
	\hat{h}(\vec{x})=\sum\limits_{i=1}^{n-1}(x_i-x_{i+1})h(S_i)+x_nh(S_n)
	\end{equation*}
	\label{definition3}
\end{myDef}
There are other forms to describe the definition of Lov{$\acute{a}$}sz extension, but in this paper, our discussion is based on this form. We should note that $\hat{h}$ is well-defined and positively homogeneous. It satisfies:
\begin{equation*}
\hat{h}({\vec1_S})=h(S), \forall S\subseteq X
\end{equation*}
where $\vec 1_S\in\{0,1\}^X$ is the characteristic vector of $S$.

There is an equivalent definition to describe the Lov{$\acute{a}$}sz extension: $\hat{h}(\vec{x})=\mathbb{E}[h(\{v_i\in X:x_i>\lambda\})]$, where $\lambda$ is uniformly random in $[0,1]$. That is,
\begin{equation*}
\hat{h}({\vec x})=\int_{\lambda=0}^{1}h(\{v_i\in X:x_i>\lambda\})d\lambda
\end{equation*}
We can understand the Lov{$\acute{a}$}sz extension $\hat{h}(\vec {x})$ as the expectation value of $h$ on a distribution $\hat{O}(\vec {x})$.  The distribution $\hat{O}(\vec {x})$ satisfies that it gets the largest subset $S_n$ with probability ${x}(v_n)$, and gets the next largest subset $S_{n-1}$ with probability ${x}(v_{n-1})-{x}(v_{n})$ and so on. We should notice that the definition of $\hat{h}(\vec{x})$ is oblivious. It does not depend on a particular function $h$.

Then we describe the process of relaxing the IMCPP. We introduce a decision variable $x_{ij}\in[0,1]$ for all $(i,j)\in M\times V$ where $x_{ij}$ is the probability that node $j$ is allocated to community $i$. Thus,
\begin{equation*}
\sum\limits_{i\in M}x_{ij}\leq 1, j\in V
\end{equation*}
The domain of the relaxed problem can be denoted as:
\begin{equation*}
P(\mathcal{M})=\{\vec x\in[0,1]^{m\times n}:\forall j\in V, \sum\limits_{i\in M}x_{ij}\leq 1\}
\end{equation*}
We use Lov{$\acute{a}$}sz extension $\hat{f}(\vec {x})$ to relax the influence function $f$ as follows:
\begin{equation}
\hat{f}(\vec {x})=\mathbb{E}_{\lambda\sim[0,1]}[f(\{(i,j)\in U: x_{ij}\textgreater \lambda\})]
\label{equation6}
\end{equation}
where $\lambda$ is uniformly random in [0, 1].

The relaxation of our problem, Equation \ref{equation2}, can be expressed as follows:
\begin{equation}
\begin{split}
&\quad\max\limits_{{x}}\hat{f}(\vec {x})\\
&\quad s.t. \  \vec x\in P(\mathcal{M})
\end{split}
\label{equation88}
\end{equation}

So we transfer our goal to maximize the Lov{$\acute{a}$}sz extension $\hat{f}(\vec {x})$ of influence function $f$ over a matroid polytope $P(\mathcal{M})$.

\begin{lemma}
	A set function $h:2^X\rightarrow R$ is submodular (or supermodular) if and only if its Lov{$\acute{a}$}sz extensions $\hat{h}$ is convex (or concave).
	\label{lemma3}
\end{lemma}
\begin{proof}
	This conclusion was shown by Jan Vondr{$\acute{a}$}k in \cite{vondrak2010continuous}.
\end{proof}
\begin{theorem}
	The relaxation  $\hat{f}(\vec {x})$ of objective function for IMCPP, shown as Equation \ref{equation6}, is monotone and concave.
\end{theorem}
\begin{proof}
	From Lemma \ref{lemma1} and Lemma \ref{lemma2}, we know that influence propagation function $f$ is monotone and supermodular. Based on Lemma  \ref{lemma3}, we have its Lov{$\acute{a}$}sz extensions $\hat{f}(\vec {x})$ is monotone and concave.
\end{proof}

\subsection{The Continuous Greedy Process}
Based on the monotone and concave property of $\hat{f}(\vec {x})$, we design a continuous process and produce a set $ \vec x\in P(\mathcal{M})$ which approximates the optimum solution $OPT=\max\{\hat{f}(\vec {x}):  \vec x\in P(\mathcal{M})\}$.  The vector moves in direction constrained by $P(\mathcal{M})$ until it achieves a local maximum gain.

In order to observe how  $\hat{f}(\vec {x})$ behaves along coordinates axes, we need to show the property of the derivative of $\hat{f}(\vec {x})$, shown as the following lemma:
\begin{lemma}
	The partial derivative for $x_i$ of the Lov{$\acute{a}$}sz extensions $\hat{h}(\vec {x})$ of a set function $h$ is
	\begin{equation*}
	\frac{\partial \hat{h}(\vec{x})}{\partial x_{i}}=h(S_i)-h(S_{i-1})
	\end{equation*}
	where $x_1\geq x_2\geq\cdots x_n$ and $S_i=\{1,2,\cdots,i\}$.
\end{lemma}
\begin{proof}
	As $h(S_0)=h(\emptyset)=0$,	based on the Definition \ref{definition3}, we have that
	\begin{equation*}
	\begin{aligned}
	\hat{h}(\vec {x})&=\sum\limits_{i=1}^{n-1}(x_i-x_{i+1})h(S_i)+x_nh(S_n)\\
	&=(x_1-x_2)h(S_1)+(x_2-x_3)h(S_2)+(x_3-x_4)h(S_3)\\&+\cdots+(x_{n-1}-x_n)h(S_{n-1})+x_nh(S_n)\\
	&=h(S_1)x_1+(h(S_2)-h(S_1))x_2+\cdots+(h(S_n)-\\&h(S_{n-1}))x_n
	\end{aligned}
	\end{equation*}
	So we can get that
	\begin{equation*}
	\frac{\partial \hat{h}(\vec{x})}{\partial x_{i}}=h(S_i)-h(S_{i-1})
	\end{equation*}
	\label{lemma4}
\end{proof}

In our IMCPP problem, we have $\vec x\in[0,1]^{M\times V}$, and let $|V|=n$, we can consider $\vec x$ as a vector such that
\begin{equation*}
\vec x=(x_{11},\cdots,x_{1n},\cdots,x_{i1},\cdots,x_{in},\cdots,x_{m1},\cdots,x_{mn})
\end{equation*}
Then, we sort this vector $\vec x$ to get its sorted vector
\begin{equation}
\vec x'=(x'_{11},\cdots,x'_{1n},\cdots,x'_{i1},\cdots,x'_{in},\cdots,x'_{m1},\cdots,x'_{mn})
\label{equation20}
\end{equation}
that satisfying $x'_{ij}\geq x'_{lk}$ if $i<l$ or $i=l\land j<k$. We denote $\Omega(x_{ij})=x'_{lk}$ and $\Omega^{-1}(x'_{lk})=x_{ij}$, which means that the element $x_{ij}$ in vector $\vec x$ corresponds to the element $x'_{lk}$ in sorted vector $\vec x'$. Let $\Gamma(x_{ij})=(i,j)$ and we have
\begin{equation*}
\begin{aligned}
S'_{lk}&=\{\Gamma(\Omega^{-1}(x'_{11})),\cdots,\Gamma(\Omega^{-1}(x'_{1n})),\cdots,\Gamma(\Omega^{-1}(x'_{l1})),\\&\cdots,\Gamma(\Omega^{-1}(x'_{lk}))\}
\end{aligned}
\end{equation*}
To compute the derivative of $\hat{f}(\vec{x})$, we have
\begin{equation}
\frac{\partial \hat{f}(\vec{x})}{\partial x_{ij}}=f(S'_{lk})-f(S'_{l(k-1)})
\label{equation9}
\end{equation}
where $\Omega(x_{ij})=x'_{lk}$.

From Equation \ref{equation9}, we can see that the derivative of $\hat{f}(\vec {x})$ for $x_{ij}$ just equals the marginal gain of influence propagation when partitioning node $j$ to community $i$ as $\Omega(x_{ij})=x'_{lk}$. So we can take advantage of this property to find a solution.

As $\hat{f}(\vec x)$ is non-decreasing monotone, for any $(i,j)\in M\times V$, $\frac{\partial \hat{f}(\vec x)}{\partial x_{ij}}\geq 0$. Thus the gradient of $\hat{f}(\vec {x})$ is a positive vector, i.e.
\begin{equation*}
\nabla \hat{f}(\vec x)={
	\left[\begin{array}{ccc}
	\frac{\partial \hat{f}(\vec x)}{\partial x_{11}},&\cdots,&\frac{\partial \hat{f}(\vec x)}{\partial x_{1n}}\\
	\vdots&\vdots&\vdots\\
	\frac{\partial \hat{f}(\vec x)}{\partial x_{m1}}&\cdots,&\frac{\partial \hat{f}(\vec x)}{\partial x_{mn}}
	\end{array}
	\right]}\geq\vec 0
\end{equation*}

Next, we begin to design the continuous process. Let $\vec x$ start from $\vec x(0)=\vec 0$ and follow a certain flow over a unit time interval:
\begin{equation*}
\frac{d\vec x(t)}{dt}=\vec v_{max}(\vec x(t)),
\end{equation*}
Then we define $\vec v_{max}(\vec x)$ as
\begin{equation*}
\vec v_{max}(\vec x(t))=\arg\max\limits_{v\in P}(v\cdot\nabla\hat{f}({\vec x(t)}))
\end{equation*}
$\vec v_{max}(\vec x)$ denotes that when an element $j$ is added to community $i$ at time $t$, the direction in which the rate of change of the tangent line of function $\hat{f}(\vec {x})$ is greatest. Based on the Equation \ref{equation9}, we know that this can bring the greatest gain for the influence propagation function $f$. As $t\in[0,1]$, we have
\begin{equation}
\vec x(t)=\int_0^t\frac{d\vec x(\tau)}{d\tau}d\tau=\int_0^t \vec v_{max}(\vec x(\tau))d\tau
\label{equation66}
\end{equation}
Next, we propose the continuous greedy algorithm for the problem, which is shown in Algorithm 	\ref{alg1}.
\begin{algorithm}[!t]
	\caption{\textbf{Continuous Greedy Algorithm}}
	\begin{algorithmic}[1]
		\Input  Graph $G$, $\mathcal{M}=(U,\mathcal{I})$, $f$
		\Output{$\vec x(1)$}
		\State Initialize $\vec x(0) = \vec 0$	
		\For {each $t\in[0,1]$}
		\State For each $(i,j)\in M\times V$, let $w_{ij}(t)=f(S'_{lk})-f(S'_{l(k-1)})$, where $\Omega(x_{ij})=x'_{lk}$
		\State $\vec v_{max}(\vec x(t))=\arg\max\limits_{\vec v\in P}(\vec v\cdot \vec w(t))$
		\State Increase $\vec x(t)$ at a rate of $\vec v_{max}(\vec x(t))$		
		\EndFor
		\State\Return $\vec x(1)$
	\end{algorithmic}
	\label{alg1}
\end{algorithm}

In this algorithm, $t$ ranges from $0$ to $1$. For each time step, we need to calculate the value of $w_{ij}(t)$. Its meaning was illustrated in Equation \ref{equation9}. The step 4 shows that $\vec v_{max}(\vec x(t))$ always equals the value which maximizes $\vec v\cdot \vec w(t)$ in every iteration. It also means that we find the maximum marginal gain value of $\hat{f}(\vec {x})$ when $\vec v(\vec x(t))\leftarrow\vec v_{max}(\vec x(t))$. Then $\vec x(t)$ increases at the rate of $\vec v_{max}(\vec x(t))$ obtained in step 4. After the for loop, we get the value of $\vec x(1)$ which is a convex combination of independent sets.

\subsection{Discrete Implementation}
Actually, the continous greedy algorithm solves our objective function by calculating the integral, shown as Equation \ref{equation66}. But it is hard to implement usually. So in this section, we discretize the continuous greedy algorithm. Given the time step $\Delta t$, the discrete version is shown as follows:

\begin{enumerate}
	\item Start with $t=0$ and $\vec x(0)=\vec 0$.
	\item Obtain $\vec w(t)$: Sort vector $\vec x(t)$ from maximum to minimum and get vector $\vec x'(t)$, shown as Equation \ref{equation20}. For each element $x_{ij}(t)$, we define $w_{ij}(t)=f(S'_{lk}(t))-f(S'_{l(k-1)}(t))$, where $\Omega(x_{ij}(t))=x'_{lk}(t)$.
	\item Let $I^*(t)$ be the maximum-weight independent set in $\mathcal{I}$ according to $\vec w(t)$.
	\item $\vec x(t+\Delta t)\leftarrow \vec x(t)+\vec1_{I^*(t)}\cdot\Delta t$.
	\item Increment $t=t+\Delta t$; if $t<1$, go back to step 2; Otherwise, return $\vec x(1)$.
\end{enumerate}
\noindent
where we denote $\nabla \hat{f}({\vec x(t)})$ by $\vec w(t)$. Because $v_{max}(\vec x(t))\in P$ and $\vec w(t)$ is non-negative, $\vec v_{max}(\vec x(t))$ corresponds to a base of matroid $\mathcal{M}$. In other words, we find a $I^*(t)\in\mathcal{I}$ such that
\begin{equation*}
I^*(t)\in\arg\max_{I(t)\in\mathcal{I}}(w(t)\cdot \vec 1_{I(t)})
\end{equation*}
where $I^*(t)$ is the maximum-weight independent set at time step $t$, which can be obtained by hill-climbing strategy. Then, $t$ increases discretely by $\Delta t$ in each step. Until getting the vector $\vec x(1)$, the algorithm terminates.

After that, we have obtain a fractional vector returned by discrete continuous greedy. Then, we take the fractional solution $\vec x(1)$ and apply randomized rounding technique: partitioning node $j$ to community $i$ with the probability $x_{ij}(1)$ independently and guaranteeing that each node can just belong to one community at most, i.e. $x_{ij}=1$ with the probability $x_{ij}(1)$ and $x_{ij}=0$ with the probability $1-x_{ij}(1)$, and for any  $j\in V$, $\sum_{i\in M}x_{ij}\leq 1$.

\section{Performance Analysis}\label{Performance}
In this section, we prove that the returned vector by Algorithm \ref{alg1} is an approximate solution of our problem, Equation \ref{equation2}. Before we get the final approximation ratio, we need to prove the following lemma.
\begin{lemma}
	Suppose there exists $\vec {x}^*\in P$ such that $OPT=\hat{f}(\vec {x}^*)$, where $OPT$ is the optimal solution for the problem, Equation \ref{equation88}. For any $\vec x\in R^n$, we have  $\vec x^*\nabla \hat{f}(\vec {x})\geq OPT-\hat{f}(\vec {x})$.
	\label{lemma5}
\end{lemma}
\begin{proof}
	We know that $\hat{f}(\vec {x})$ is concave in all directions, but we just need to consider the non-negative direction as $\vec {x}^*\in P$.  Let us consider a direction $\vec d=(\vec x^*-\vec x)\vee\vec 0$, where $x\vee y=max(x,y)$. Assume that $\vec d\geq \vec0$ is the moving direction of $\hat{f}(\vec {x})$. Hence, based on the property of concave function, we  have the conclusion: $\hat{f}(\vec {x}+\vec d)-\hat{f}(\vec {x})\leq \vec d\nabla \hat{f}(\vec {x})$.
	
	We discuss the problem in two cases:
	
	1. $\vec x+\vec d=\vec x^*\vee \vec x \geq \vec x^*$. As $\hat{f}(\vec {x})$ is non-decreasing, we can get $\hat{f}(\vec {x}+\vec d)\geq \hat{f}({\vec x^*})$. So $\hat{f}(\vec {x}+\vec d)-\hat{f}(\vec {x})\geq \hat{f}({\vec x^*})-\hat{f}(\vec {x})$.
	
	2. When $\vec x^*-\vec x\geq\vec 0$, so $\vec d=(\vec x^*-\vec x)$, $\vec d\leq \vec x^*$ and as $\nabla \hat{f}(\vec {x})\geq\vec0$, we have $\vec d\nabla \hat{f}(\vec {x})\leq \vec x^*\nabla\hat{f}(\vec {x})$.
	
	Combining the above two cases, we  can obtain that:
	\begin{equation*}
	\vec x^*\nabla \hat{f}(\vec {x})\geq OPT- \hat{f}(\vec {x})
	\end{equation*}
	
\end{proof}
\begin{theorem}
	When $\hat{f}({x})$ is the Lov{$\acute{a}$}sz extension of the influence propagation $f$ for IMCPP, $\vec x(1)$ returned by Algorithm \ref{alg1} satisfies:
	$\vec x(1)\in P$ and $\hat{f}({\vec x(1)})\geq (1-\frac{1}{e})OPT$
	\label{theorem1}
\end{theorem}
\begin{proof}
	Based on the Equation \ref{equation66}, we can calculate $\vec x(1)$ as follows:
	\begin{equation*}
	\vec x(1)=\int_{0}^{1}\vec x^\prime(t)dt=\int_{0}^{1}\vec v_{max}(\vec x(t))dt
	\end{equation*}
	Since $t$ ranges from 0 to 1, we can use the theorem of the limit of Riemann sum to calculate the limit of the integral of $\vec x(1)$.
	\begin{equation*}
	\vec x(1)=\lim\limits_{n\rightarrow\infty}\frac{1}{n}\sum_{i=1}^{n}\vec v_
	{max} \vec x(\frac{i}{n})
	\end{equation*}
	By the definition of $\vec v_{max} (\vec x)$, we know that $\vec v_{max} (\vec x)\in P$ for any $\vec x$. The term inside the limit is a convex combination of vectors which belongs to $P$. Since $P$ is a closed convex set, the limit of $\vec x(1)$ is in $P$. This proves $\vec x(1)\in P$.
	
	From Lemma 	\ref{lemma5}, we know that there exists $\vec v\in P$ such that $\vec v({x})\nabla \hat{f}(\vec{x})\geq OPT-\hat{f}(\vec{x})$. When $\vec v=\vec v_{max}$, we also can get $\vec v_{max}({\vec x})\nabla \hat{f}({\vec x})\geq OPT-\hat{f}({\vec x})$.
	
	We get further $\frac{d\hat{f}({\vec x(t)})}{dt}\geq OPT-\hat{f}({\vec x(t)})$. We define $g(t)=\frac{d\hat{f}({\vec x}(t))}{dt}+\hat{f}({\vec x}(t))\geq OPT$.  Then we can get $\hat{f}({\vec x}(t))=\int_0^te^{(x-t)}g(x)dx$. So $\hat{f}({\vec x}(1))=\int_0^1e^{(x-1)}g(x)dx\geq\int_0^1e^{(x-1)}OPTdx=OPT[e^{x-1}]^1_0=OPT(1-1/e)$. This proves $\hat{f}({\vec x(1)})\geq (1-\frac{1}{e})OPT$.
\end{proof}

In the second stage of Algorithm \ref{alg1}, we use random rounding to convert the fractional solution to integer solution. As we know that the relationship between the result of random rounding $\hat{f}_R({\vec x})$ and the continuous solution $\hat{f}({\vec x})$  is  $\hat{f}_R({\vec x})\geq \hat{f}({\vec x})$. So the final result of Algorithm \ref{alg1} we present compared with the optimal solution is $\hat{f}_R({\vec x})\geq (1-\frac{1}{e})OPT$.

\begin{theorem}
	Algorithm \ref{alg1} returns a $(1-\frac{1}{e})$-approximation (in expectation) for the problem, Equation  \ref{equation88}.
	\label{theorem2}
\end{theorem}
Theorem	\ref{theorem1} and the proof above imply the result of Theorem \ref{theorem2}.

Here we will discuss the complexity of the proposed algorithm. The complexity is relatively high for large scale social networks.
\begin{theorem}
	The complexity of discrete continuous greedy algorithm is upper bounded by $O((\log(mn) + mn|E|r)/\Delta t)$.
	\label{theorem3}
\end{theorem}
\begin{proof}
	First, at step (2), we sort the elements in $\vec x(t)$ from maximum to minimum, there are total $m\times n$ elements, the complexity is $O(\log(mn))$. Then, we estimate the objective function $f(\cdot)$ by Monte Carlo simulations, the running time of $f(\{v\})$ given a node $v$ is $O(|E|r)$ where $r$ is the number of Monte Carlo simulations, where $|E|$ is the cardinality of the edge set $E$ of the social network. The average number of node in $S'_{lk}(t)$ is $mn/2$, thus, the total running time of step (2) is $O(\log(mn) + mn|E|r)$.
	
	The running time of Discretized continuous greedy is determined by its step (2), so we have its time complexity $O((\log(mn) + mn|E|r)/\Delta t)$
\end{proof}

 We can know that the  complexity is high from Theorem \ref{theorem3}, this results the poor scalability of the network.

\section{Experiments}\label{experiments}
We show the experimental results of the proposed algorithms is this section. Firstly, we give descriptions about the used datasets and parameter settings.
\subsection{Experimental Setup}
\textbf{Datasets}: Three datasets of different magnitude are used in our experiments. The first two datasets are from networkrepository.com, it is like an online network repository which includes diverse kinds of networks. The first dataset is a co-authorship network which is a co-authorship about scientists in the field of network theory and experiment. The dataset consists of nodes and edges which depict relationships among 379 users, it has 914 edges between nodes; the second dataset is a Wiki-vote network, namely Wikipedia who-votes-on-whom network. It shows the voting relationship among 914 users and the edges between users are 2914. The third dataset is from arXiv, it is called the NetHEPT, which is a co-authorship network in the ``High Energy Physics'' section. It is made by 15299 users and 31376 edges.

\textbf{Influence Model}: LT model is the influence spread model in our problem. We set the spread probability for each directed edge $e$ as $p(e)=1/d(i)$, where $d(i)$ is the in-degree of node $i$. This method of setting $p(e)$ has been widely adopted in some previous research\cite{yang2016continuous, wang2012scalable, tang2014influence}. The threshold that a node becomes active is generated randomly between 0 and 1.

\textbf{Comparison Methods}: To evaluate the effectiveness of the proposed algorithm, we compare the discrete continuous greedy algorithm with two baseline algorithms: random method, label propagation algorithm. Besides the Spit algorithm for Maximum K-Community Partition (SAMKCP) algorithm and Merge algorithm for Maximum K-Community Partition (MAMKCP) which are described in \cite{lu2014influence} are also used as the comparision algorithms.

\textbf{Random}: It randomly partitions nodes to different communities, which is a classical baseline algorithm.

\textbf{Label Propagation} \cite{garza2019community}: It is a classical disjoint community detection algorithm. Firstly, each node is given a unique label. Then these labels spread in the network. Each node updates the label adopted by most of its neighbors iteratively. At the end of the algorithm, connected nodes with the same label make up one community.

\textbf{SAMKCP}: All the nodes belong to one community at first, then they spits on one of the communities recursively, which is a heuristic algorithm.

\textbf{MAMKCP}:  Each node belongs to a community, then pairs of communities are merged recursively as a new community, which is also a heuristic algorithm.
\subsection{Result Analysis}
We have to extract sub-graph firstly at each step of the experiment so as to estimate the influence propagation of a community. The process of extracting sub-graph is that: Given the node sets of  communities, we go through all the edges. We add a edge to the set of edges of sub-graph when the two nodes of this edge are in the node set of this community. Then do simulations in the next steps. To estimate $f(\cdot)$, the number of Monte Carlo simulation is set as 500.

\begin{figure}[htbp]
	\centerline{\includegraphics[width=4.5cm, height=2.3cm]{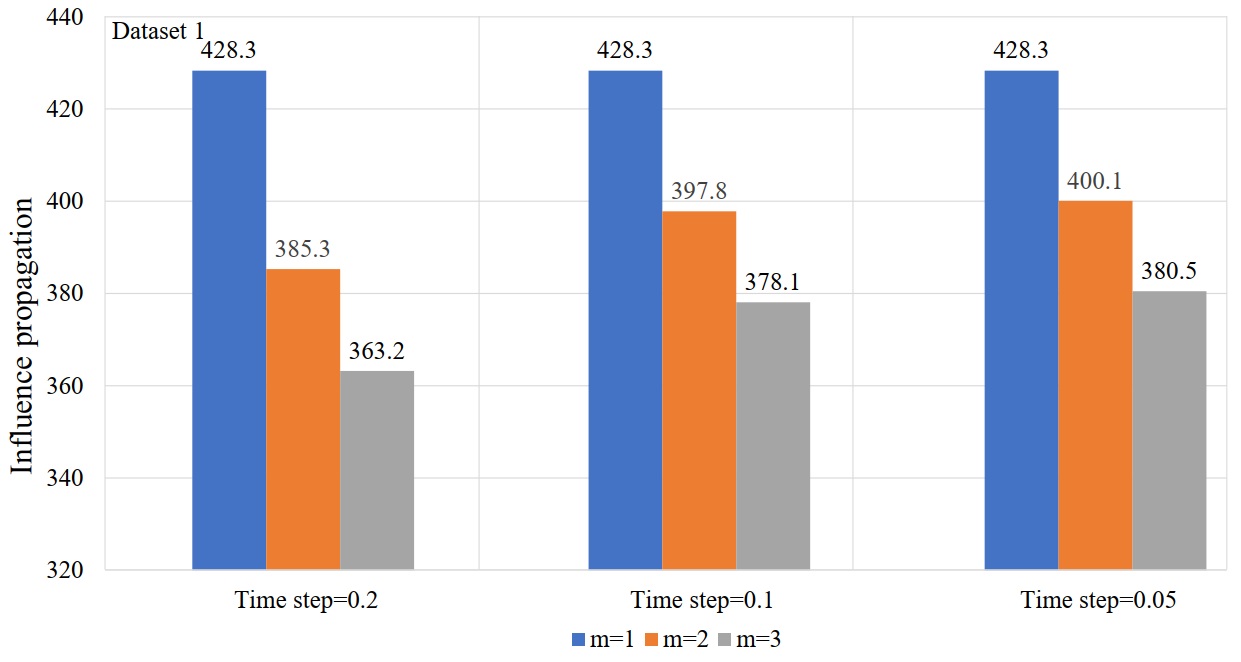}}
	\caption{Continuous results on dataset 1}\label{dataset1}
\end{figure}

\begin{figure}[htbp]
	\centerline{\includegraphics[width=4.5cm, height=2.3cm]{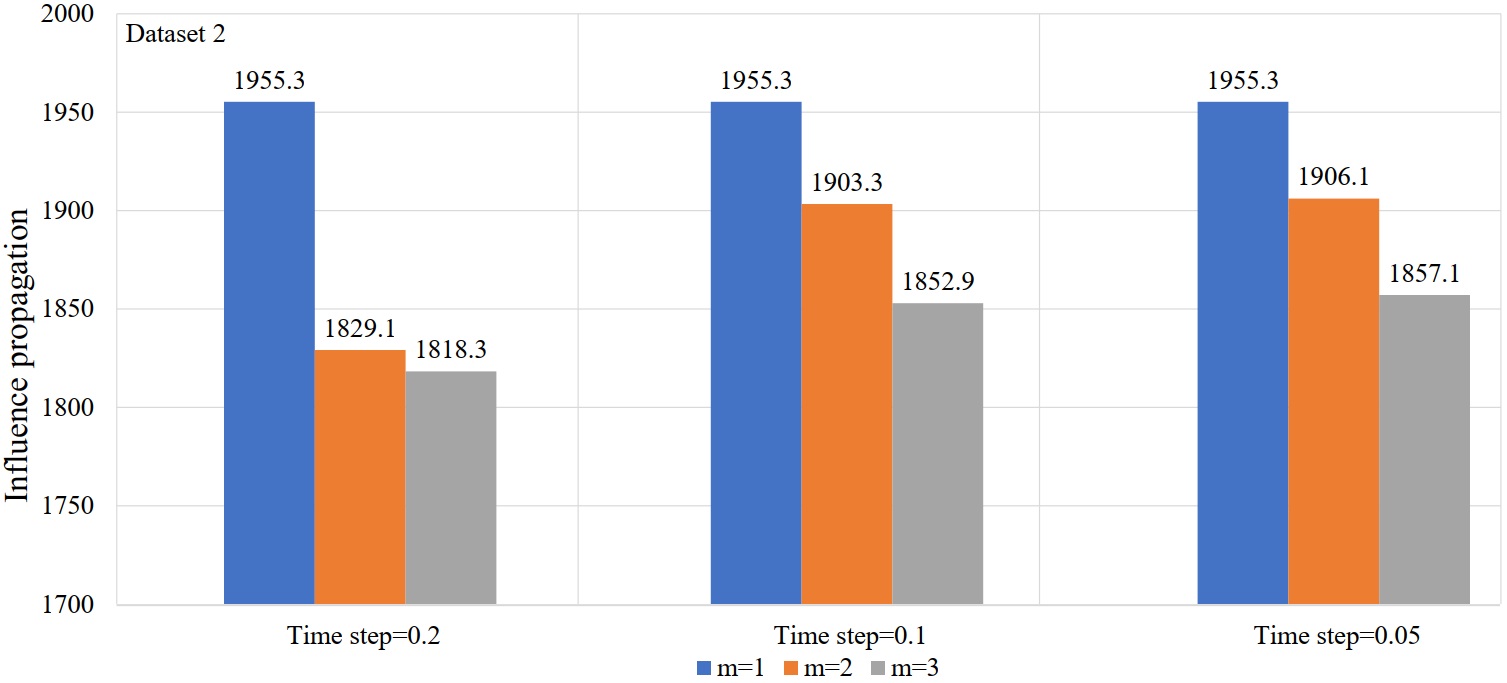}}
	\caption{Continuous results on dataset 2}\label{dataset2}
\end{figure}

\begin{figure}[htbp]
	\centerline{\includegraphics[width=4.5cm, height=2.2cm]{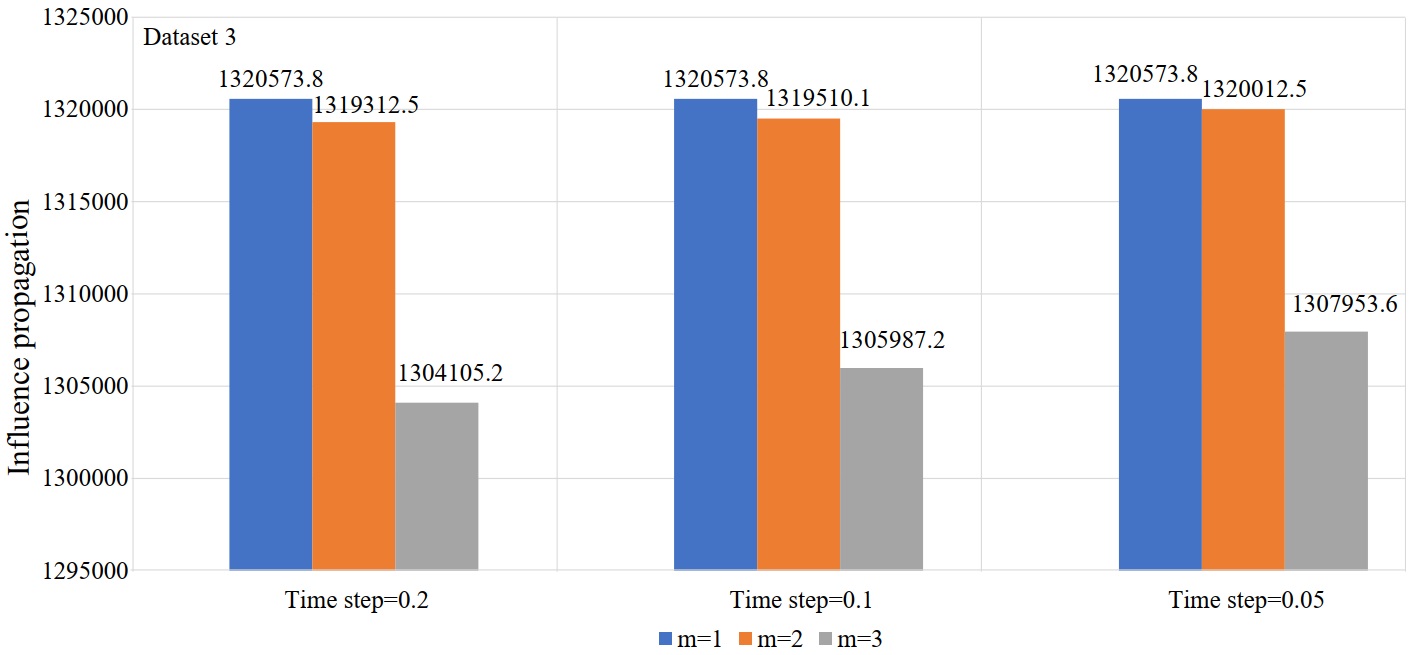}}
	\caption{Continuous results on dataset 3}\label{dataset3}
\end{figure}

\textbf{Varying the value of $m$ and the datasets}. The experiment result in Figure \ref{dataset1} is done on dataset 1, we show the variation of total influence propagation after the community partitioning by changing the value of time step $\Delta t$ when we partition the community to different numbers with the  method of Discrete continuous greedy. We can see from the bars that when $m=1$ the influence propagation does not change no matter what $\Delta t$ is since we do not need to partition community actually. When $m=2$ the influence propagation is always larger than that in $m=3$ whatever the time step $\Delta t$ is. When $m=1$, the influence propagation obtains the maximum value. So the larger the number of the community is, the smaller the total influence propagation is in the social network. But if we fix the number of community partition $m$, we can see that when the time step $\Delta t$ decreases from 0.2 to 0.05, the influence propagation is increasing from 385.3 to 400.1 when $m=2$ and increasing from 363.2 to 380.5 when $m=3$. Thus, we can get that the smaller the value of time step $\Delta t$ is, the larger the total influence propagation after community partitioning is, so smaller time step $\Delta t$ can make the community partition more accurate. We can also see from the Figure \ref{dataset1} that when the time step $\Delta t$ decreases from 0.2 to 0.1 and $m=2$, the increment of influence propagation is 12.5; when we just increase the  number of community partition to $m=3$, the increment of influence propagation increases to 14.9. But when the time step $\Delta t$ decreases from 0.1 to 0.05, the increment of influence propagation is 2.3 when $m=2$ and 2.4 when $m=3$, respectively. The increment is greatly deuced as the decreasing of $\Delta t$. When time step $\Delta t=0.1$ the influence propagation is already relatively stable. These show that the community partition is already fairly accurate when $\Delta t=0.1$ although a smaller the value of $\Delta t$ will get a more accurate result. The total influence propagation will become smaller if we try to partition the whole network into more communities. This phenomenon is because that it decreases the leaking out of the influence propagation between two communities after the community partitioning.

The experiment results in Figure  \ref{dataset2} and \ref{dataset3} are done on two larger dataset 2 and dataset 3.  we can find that the trend of the bar charts in Figure \ref{dataset2} and \ref{dataset3} are similar to the results in Figure \ref{dataset1}. This further confirms the correctness and validity of our results.

\begin{figure}[htbp]
	\centerline{\includegraphics[width=4.5cm, height=2.3cm]{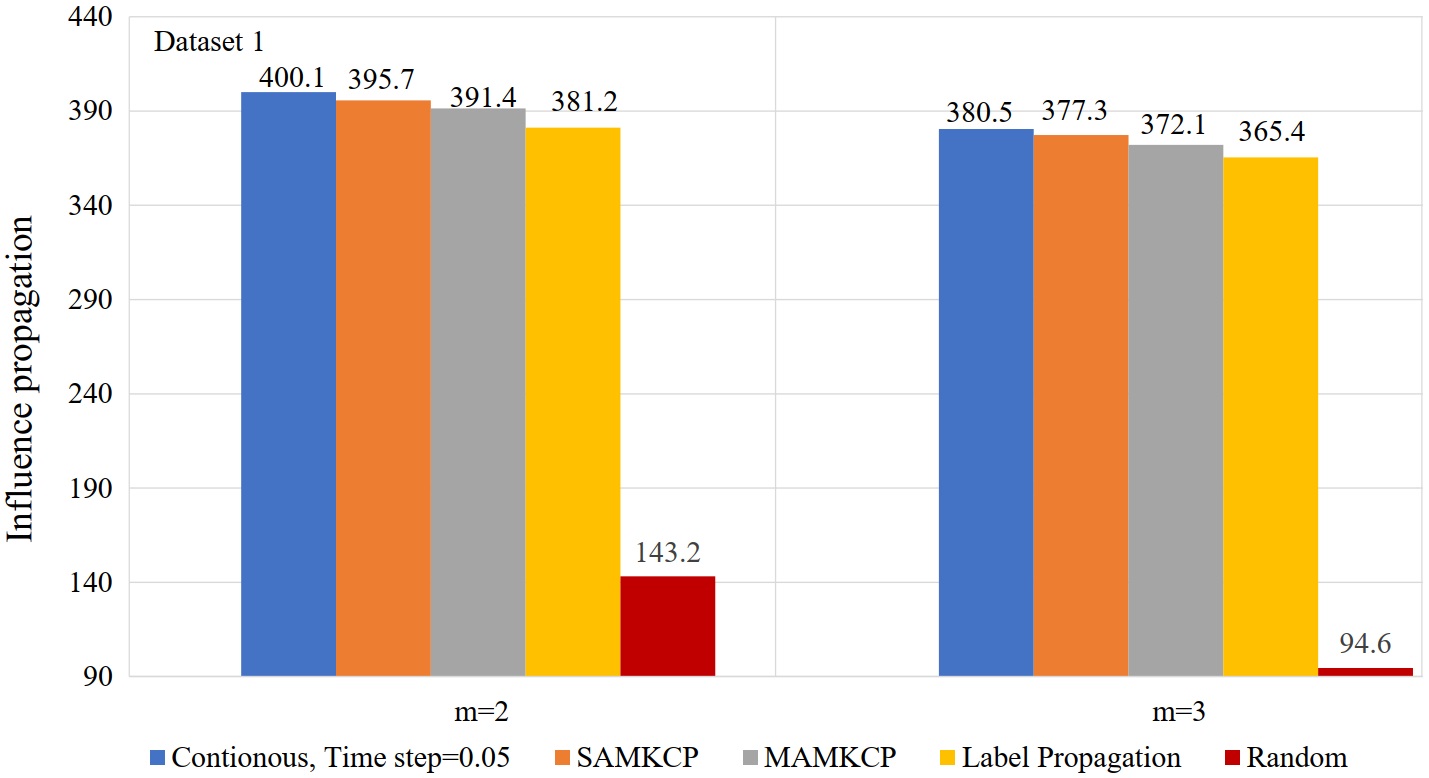}}
	\caption{Comparative results on dataset 1}\label{dataset11}
\end{figure}

\begin{figure}[htbp]
	\centerline{\includegraphics[width=4.5cm, height=2.3cm]{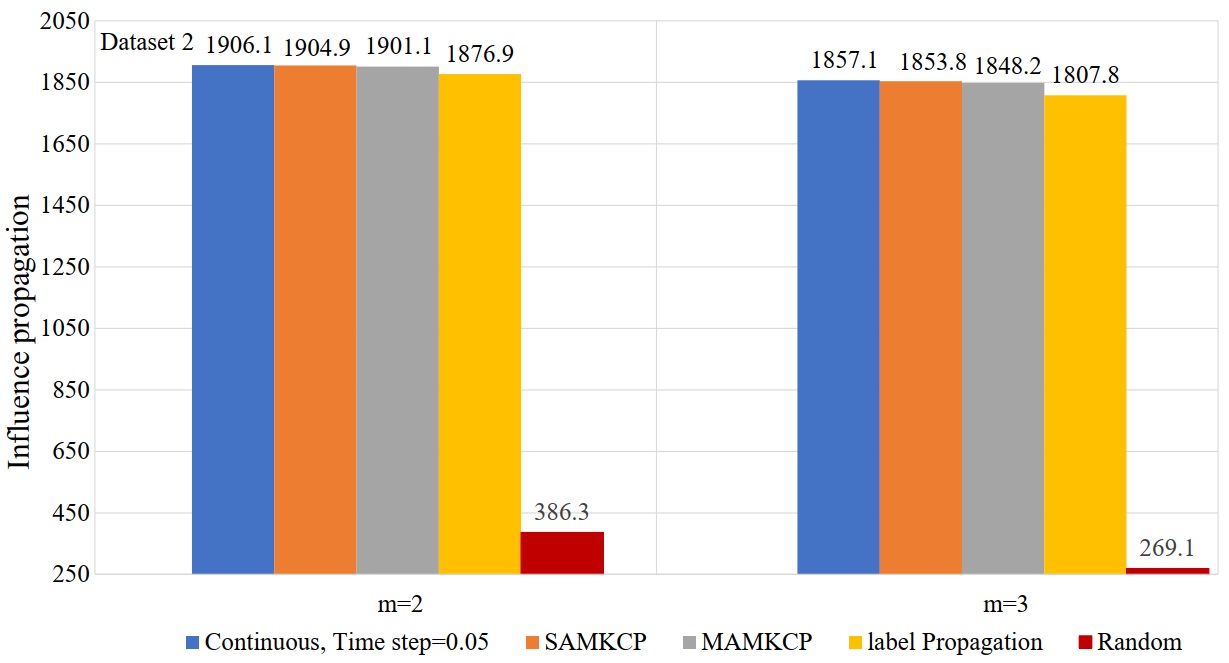}}
	\caption{Comparative results on dataset 2}\label{dataset22}
\end{figure}

\begin{figure}[htbp]
	\centerline{\includegraphics[width=4.5cm, height=2.3cm]{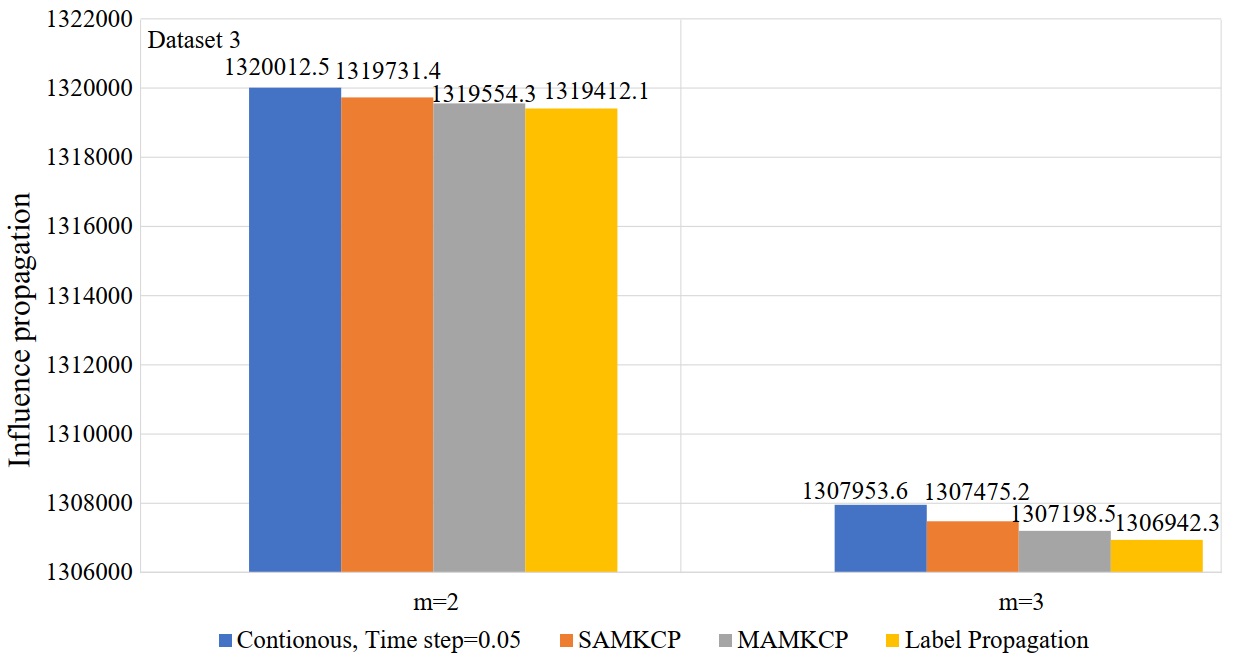}}
	\caption{Comparative results on dataset 3}\label{dataset33}
\end{figure}

\textbf{Comparison with other methods}. In order to show the effectiveness of our approach, we compare Discrete continuous greedy with Random community partition method, Label propagation method and other two methods SAMKCP, MAMKCP proposed by \cite{lu2014influence}. The comparison results are shown in Figure \ref{dataset11}, \ref{dataset22} and \ref{dataset33}. The random community partition method is intuitive. It simple partitions nodes randomly to each community. The y-axis shows the influence propagation when we partition different number of communities on dataset 1, 2 and 3 using our algorithm and random method, label propagation method,  SAMKCP, MAMKCP. The time step is set to 0.05 in proposed algorithm. It is obvious that the influence propagation with our proposed algorithm is much more than random method no matter $m=2$ or $m=3$ and no matter in dataset 1 or dataset 2. As in dataset 3, the result of random community partition method is much worse than the other four methods, we omit its result. The result of our method is also superior to SAMKCP, MAMKCP and Label propagation method, which shows that our algorithm trades time complexity for more accurate performance than SAMKCP, MAMKCP and Label propagation method. SAMKCP, MAMKCP and Label propagation methods have a low computational complexity but also have some loss in performance. These results illustrate that the continuous greedy method we proposed works well to maximize the influence propagation in each community.

\section{Conclusion}\label{conclusion}
We use the Lov{$\acute{a}$}sz extension theory to relax our target function of the Influence Maximization for Community Partition Problem (IMCPP) and introduce a partition matroid to the domain of the relaxed problem. Then we propose a continuous greedy algorithm and its discrete form to solve the problem. In order to convert the fractional solutions which are obtained by the two algorithms to integer solutions, we use the random rounding method to the results of the algorithms at the second stage.  We analyze the performance of our proposed algorithm and then an $1-1/e$ approximation ratio is got for the algorithm. Finally, we do simulations on three datasets which is from the real-world social networks to show the advantages of the proposed methods.

In the future, we would like to design an approximation algorithm for the influenced-based community partition problem in IC model with the sandwich method or DS decomposition method. And based on these work, we want to study a two stage algorithm for rumor blocking: doing the influence-based community partition at the first stage and blocking rumor at bridge-ends of the community which contains the rumor source at the second stage.
%\section*{Acknowledgment}
%This work is supported by the National Natural Science Foundation of China (No.61772385, No.61572370).
\bibliographystyle{IEEEtran}
\bibliography{ref}
\end{document}